\documentclass[runningheads,a4paper]{llncs}

\usepackage{amssymb}
\usepackage{graphicx}

\usepackage{url}
\urldef{\mailsa}\path|{alfred.hofmann, ursula.barth, ingrid.haas, frank.holzwarth,|
\urldef{\mailsb}\path|anna.kramer, leonie.kunz, christine.reiss, nicole.sator,|
\urldef{\mailsc}\path|erika.siebert-cole, peter.strasser, lncs}@springer.com|    
\newcommand{\keywords}[1]{\par\addvspace\baselineskip
\noindent\keywordname\enspace\ignorespaces#1}

\usepackage{amsmath}
\usepackage{epstopdf}

\usepackage{caption}
\usepackage{subcaption}
\usepackage{framed}
 \usepackage[usenames,dvipsnames]{pstricks}
 \usepackage{pst-grad} 
 \usepackage{pst-plot} 

\begin{document}

\mainmatter  

\title{Decomposing Truthful and Competitive Online Double Auctions}

\titlerunning{Truthful and Competitive Online Double Auctions}

%
%

\author{Dengji Zhao\inst{1} \and Dongmo Zhang\inst{2} \and Laurent Perrussel\inst{3}
}
\authorrunning{Zhao \textit{et al}.}


\institute{
	Graduate School of ISEE,
	Kyushu University, Japan\\
	djzhao@inf.kyushu-u.ac.jp
\and
	ISL,
	University of Western Sydney, Australia\\
	d.zhang@uws.edu.au
\and
   	IRIT, University of Toulouse, France\\
	laurent.perrussel@univ-tlse1.fr
}

%
%

\toctitle{Lecture Notes in Computer Science}
\maketitle

\begin{abstract}
In this paper, we study online double auctions, where multiple sellers and multiple buyers arrive and depart dynamically to exchange one commodity. We show that there is no deterministic online double auction that is truthful and competitive for maximising social welfare in an adversarial model. However, given the prior information that sellers are patient and the demand is not more than the supply, a deterministic and truthful greedy mechanism is actually $2$-competitive, i.e. it guarantees that the social welfare of its allocation is at least half of the optimal one achievable offline. Moreover, if the number of incoming buyers is predictable, we demonstrate that an online double auction can be reduced to an online one-sided auction, and the truthfulness and competitiveness of the reduced online double auction follow that of the online one-sided auction. Notably, by using the reduction, we find a truthful mechanism that is almost $1$-competitive, when buyers arrive randomly. Finally, we argue that these mechanisms also have a promising applicability in more general settings without assuming that sellers are patient, by decomposing a market into multiple sub-markets. 
\end{abstract}

\section{Introduction}
Double auction markets (aka exchanges) allow multiple sellers and multiple buyers to trade a commodity simultaneously, e.g. the New York Stock Exchange. Each trader (seller or buyer) has a private valuation of the commodity. In a double auction market, sellers submit asks (sell orders) 
to sell a commodity and buyers submit bids (buy orders) 
to buy the commodity. We assume that each trader supplies/demands one unit of the commodity. Given asks and bids submitted by traders, the market owner (aka auctioneer) matches them using certain market clearing polices in order to make transactions. Traditionally, double auctions have been well studied in static settings, where all traders are known before the auctioneer makes any decision~\cite{Myerson_1983,McAfee_1990,Wurman_1998}. However, in most modern double auction markets, traders arrive and depart at different times. We call these markets \textit{online double auctions}. The main challenge for the auctioneer in an online double auction is to make decisions without knowing the traders/orders not yet arrived. The decisions involve an online bipartite matching (i.e. allocation) between sellers and buyers and a payment calculation.

Following the previous work in online auction design~\cite{blum_online_2006,Bredin_2007,Parkes_OnlineMD_2007}, this article makes an incremental step in this field. We focus on two important criteria, \textit{truthfulness} and \textit{efficiency}, for online double auction design. We say a mechanism is truthful if for each trader, reporting his truthful type, including valuation, arrival and departure time, is his dominant strategy, and it is efficient if the social welfare of its allocation is maximised among all feasible allocations.

\subsection{Contributions}
We show that there is no deterministic and truthful online double auction that is also competitive for efficiency in an adversarial model. Then we further study the environment where sellers are relatively static compared with buyers. Within this environment, two situations are examined: 1) the demand (the number of buyers) is not more than the supply (the number of sellers), but not known exactly, 2) the demand is predictable but not necessarily not more than the supply. We show that, in the first situation, a deterministic and truthful greedy-mechanism is actually $2$-competitive. In the second situation, we propose a framework to reduce a truthful online double auction to a truthful online one-sided auction, and demonstrate that the competitiveness of the reduced online double auction follows that of the online one-sided auction. Especially, by using the reduction, we achieve a truthful auction that is almost $1$-competitive, i.e. the social welfare of the auction's allocation is nearly optimised, when buyers arrive randomly. Finally, we show that the assumption made on sellers' arrival and departure can be relaxed by, for example, decomposing a market into multiple disjoint sub-markets and applying the proposed mechanisms in each sub-market.

\subsection{Related Work}
During last decade, there have been substantial researches on mechanism design in different dynamic environments, termed \textit{online mechanism design} (see \cite{Parkes_OnlineMD_2007} for a survey). Most of the previous work has focused on one-sided dynamic markets where either the supply or the demand is dynamic, e.g. Ad auctions~\cite{Mehta:2007}. More importantly, the auctioneer (in most cases, the seller) in one-sided dynamic markets does not provide valuations (or reserve prices) to the commodities exchanged and is not considered to strategically manipulate the auction. However, in online double auction markets, both the supply and the demand are dynamic and playing strategically, and the auctioneer has no control of any of them. 

To tackle the complexity of online double auction design, existing research has utilised certain accessible prior knowledge of the dynamics to design desirable online auctions~\cite{blum_online_2006,Bredin_2007}. 
For instance, given the assumption that the valuations of traders are in a range $[p_{min}, p_{max}]$, Blum \textit{et al.}~\cite{blum_online_2006} proposed a $r$-competitive truthful online double auction in an adversarial setting for maximising social welfare, where $r$ is the fixed point of $r = \frac{1}{2}\ln\frac{p_{max}-p_{min}}{(r-1)p_{min}}$. Besides that, they also considered many other criteria. Moreover, assumed that traders' available/active time period in the auction is no more than some constant $K$, Bredin {\it et al.}~\cite{Bredin_2007} designed a framework to construct truthful online double auctions from truthful static double auctions, and demonstrated the performance (for maximising social welfare) of the auctions given by the framework in probabilistic settings through experiments. 

This paper is organised as follows. In Section \ref{sect_pre}, we briefly introduce our market model and related concepts. In Section \ref{sect_neg}, we show the impossibility result. Then we propose a deterministic and truthful mechanism that is $2$-competitive in Section \ref{sect_case1} and a framework to reduce a truthful online double auction to a truthful online one-sided auction in Section \ref{sect_case2} for two restricted environments respectively. We conclude in Section \ref{sect_con} with some discussions.

\section{Preliminaries and Notations}
\label{sect_pre}
We consider an online/dynamic double auction market, in which a set $B$ of {\bf buyers} and a set $S$ of {\bf sellers} trade one commodity. Buyers and sellers are {\bf traders}. We will refer to a seller as she and a buyer or trader as he. Let $T=B\cup S$ and assume that traders are independent and no trader can be both buyer and seller at the same time, i.e. $B\cap S = \emptyset$. Each trader supplies or demands a single unit of the commodity during a specific time period called the active time of the trader. Since each trader might have different active times, they might come and leave the market at different times, which causes the dynamics of the market. Given the dynamics of the market, the \textbf{auctioneer} (market owner) is challenged by making decisions without knowing those traders not yet arrived.

Each trader $i\in T$ has a privately observed {\bf type} $\theta_i=(v_i, a_i,d_i)$, where $v_i, a_i, d_i \in \mathbb{R}^+$, $v_i$ is $i$'s valuation of a single unit of the commodity, and $a_i$ and $d_i$ are the starting point and the ending point of $i$'s \textbf{active time}, i.e. the arrival and departure time of $i$. 

Due to the revelation principle~\cite{Myerson:2008}, we focus on mechanisms that require traders to directly report their types. 
However, traders do not necessarily report their true types but no early-arrival and no late-departure misreports are permitted, i.e. given trader $i$'s type $\theta_i=(v_i, a_i,d_i)$, his report $\theta_i^\prime=(v_i^\prime, a_i^\prime, d_i^\prime)$ satisfies $a_{i}^\prime \leq d_{i}^\prime$ and $[a_{i}^\prime, d_{i}^\prime] \subseteq [a_i, d_i]$. The intuition behind this constraint is that traders do not recognise the market before their arrival and they do not get utility for any trade happened after their true departure. We say a seller's report (called \textbf{ask}) $\theta_i=(v_i,a_i,d_i)$ and a buyer's report (called \textbf{bid}) $\theta_j=(v_j,a_j,d_j)$ are \textbf{matchable} if and only if $v_i\leq v_j$ and $[a_i,d_i]\cap [a_j, d_j]\not=\emptyset$. 
That is, a match/transaction should not decrease social welfare.

Let $\theta = (\theta_i)_{i\in T}$ denote a complete type profile, and $\theta^A = (\theta_i)_{i\in S}$ and $\theta^B = (\theta_i)_{i\in B}$ be the complete ask and bid profile respectively. Let $\theta_{-i}$ be the type profile of all traders except for $i$.
\begin{definition}
 An \textbf{online double auction (ODA)} $\mathcal{M} = (\pi, x)$ consists of an \textbf{allocation policy} $\pi=(\pi_i)_{i\in T}$ and a \textbf{payment policy} $x = (x_i)_{i\in T}$, where $\pi_i(\theta) \in \{0,1\}$ indicates whether or not trader $i$ trades successfully during his reported active time ($1$ means success), and $x_i(\theta)\in \mathbb{R}_+$ determines the payment paid (received) by buyer (seller) $i$ during his (her) reported active time.
\end{definition} 

An allocation $\pi$ is \textbf{feasible} if $\sum_{i\in B}\pi_i(\theta) = \sum_{i\in S}\pi_i(\theta)$ for all $B$, $S$ and $\theta$. An ODA $\mathcal{M} = (\pi, x)$  is feasible if $\pi$ is feasible. Feasibility guarantees that the auctioneer never takes short or long position in the commodity exchanged in the market. Only feasible ODAs will be discussed in this article.  

Given trader $i$ of type $\theta_i=(v_i, a_i, d_i)$, report profile $\theta^\prime$ and ODA $\mathcal{M} = (\pi, x)$, let $v(\theta_i) = v_i$, and the \textbf{utility} of $i$ is defined as
\begin{equation*}
u(\theta_i,\theta^\prime,(\pi,x)) = \left\{
   \begin{array}{ll}
    v(\theta_i)\pi_i(\theta^\prime)-x_i(\theta^\prime), & \text{if $i \in B$.}\\
    x_i(\theta^\prime) - v(\theta_i)\pi_i(\theta^\prime), & \text{if $i \in S$.}\\
   \end{array}
 \right.
\end{equation*}

\begin{definition}
 An ODA $\mathcal{M} = (\pi, x)$ is \textbf{truthful} (aka \textbf{incentive-compatible}) if 
 $u(\theta_i,(\theta_i,\theta_{-i}^\prime),(\pi,x)) \geq u(\theta_i,\theta^\prime,(\pi,x))$ for all $i$, all permitted misreports $\theta^\prime$ of $\theta$, all type profile $\theta$. 
\end{definition}
\begin{definition}
An ODA $\mathcal{M} = (\pi, x)$ is \textbf{efficient} if $\mathcal{M}$ maximises the  \textbf{social welfare}
\begin{equation}
W(\pi(\theta)) = \sum_{i\in B}v(\theta_i)\cdot \pi_i(\theta) + \sum_{i\in S}v(\theta_i)\cdot (1 - \pi_i(\theta))
\end{equation}
for all type profile $\theta$. 
\end{definition}

In other words, an ODA is efficient if it always allocates items to those traders who value them most highly. In a market with dynamic participants, it is often not possible for an online mechanism to guarantee efficient allocations without the knowledge of the dynamics, because the mechanism's decision-making is challenged by the uncertainty of future participants. Therefore, we measure an online mechanism's efficiency by competitive analysis, namely, we compare the social welfare obtained by an online mechanism with the maximal social welfare one can achieve offline, i.e. when the mechanism knows all future coming reports. Given type profile $\theta$, let $Opt(\theta)$ be the \textbf{optimal allocation} giving the optimal/maximal social welfare. Note that $Opt(\theta)$ is also constrained by feasibility. The following notion of competitiveness will be used to measure the efficiency of ODAs.
\begin{definition}
\label{def_competitive}
An ODA $\mathcal{M}=(\pi, x)$ is \textbf{$c$-competitive} if for any type profile $\theta$, the  social welfare of $\pi(\theta)$
$W(\pi(\theta)) \geq \frac{W(Opt(\theta))}{c}$.
We refer to $c$ as the \textbf{competitive ratio} of $\mathcal{M}$ for efficiency.
We say that $\mathcal{M}$ is \textbf{competitive} if $\mathcal{M}$ is $c$-competitive for some constant $c>0$.
\end{definition}

Moreover, we say a mechanism is \textbf{individually rational} if it gives its participants non-negative utility, i.e. they are not forced to participate, and a mechanism is \textbf{budget balanced} if the mechanism receives zero profit or \textit{weakly budget balanced} if its profit is non-negative. All the mechanisms discussed in the rest are individually rational without further mention.

Note that the mechanism we defined in the above are \textit{deterministic}. Non-deterministic mechanisms are also proposed/discussed in the rest, and they can be represented as a probabilistic combination of deterministic ones.

\section{No Deterministic ODA is Universally Competitive}
\label{sect_neg}
In this section, we will demonstrate that no deterministic and truthful ODA is competitive in an adversarial model. That is, for any deterministic and truthful ODA $\mathcal{M} = (\pi, x)$, there exists a type profile $\theta$ such that the social welfare $W(\pi(\theta))$ is infinitely far from the optimal one $W(Opt(\theta))$.

\begin{theorem}
\label{the_nag}
 For any deterministic and truthful ODA $\mathcal{M} = (\pi, x)$ and any $c>0$, there exists a type profile $\theta$ such that $W(\pi(\theta)) \leq \frac{W(Opt(\theta))}{c}$.
\end{theorem}
\begin{proof}
A deterministic ODA makes decisions at a bid's/ask's arrival time, departure time and/or predefined time points. 

If decisions are not only made at asks' departure time, then we can always find a type profile $\theta^\prime$ such that the last arrived ask $\theta_{last}$ of $\theta^\prime$ is matched by $\mathcal{M}$ before $\theta_{last}$ departs. Let $\theta = (\theta^\prime, \theta_{*})$ where $\theta_* = (v_*, a_*, d_*)$ is a bid and it arrives after $\theta_{last}$ is matched and before $\theta_{last}$ departs. Since $\mathcal{M}$'s decision does not depend on traders not yet arrived, $\theta_*$ will not be matched by $\pi(\theta)$ because there is no unmatched ask available. There exists a $\theta_*$ such that $\theta_*$ is matched by $Opt(\theta)$ (if $v(\theta_*)$ is sufficiently large) and $W(\pi(\theta)) \leq \frac{v(\theta_*)}{c} \leq \frac{W(Opt(\theta))}{c}$. Therefore, if $v(\theta_*)$ approaches to $\infty$, $c$ will also approach to $\infty$.

Otherwise, i.e. decisions are only made at asks' departure time, there exists a type profile $\theta$ where the last arrived bid $\theta_* = (v_*, a_*, d_*)$ arrives after the second last ask's departure, and departs after the last ask's arrival but before the last ask's departure, where we also get $W(\pi(\theta)) \leq \frac{v(\theta_*)}{c} \leq \frac{W(Opt(\theta))}{c}$ if $v(\theta_*)$ is sufficiently large. Note that truthfulness is necessary to guarantee that all types are truthfully reported so that  social welfare is correctly measured. \qed
\end{proof}

Given the above impossibility, we can still search for non-deterministic and competitive mechanisms or examine cases where the dynamics is limited by, say, certain prior knowledge of the future participants. For instance, we may know the total number of traders arriving in the future, or traders' valuation satisfying some known distributions. With certain prior knowledge of the traders, ODAs with desirable properties are achievable, e.g.~\cite{blum_online_2006,Bredin_2007}.

In the rest of this paper, we further study two environments with prior information. In both cases, we assume that sellers are \textbf{patient}, i.e. they are active before the first buyer's arrival until the arrival of the last buyer. In the first case, we further assume that the demand is no more than the supply, while in the other case we assume that we know how many buyers will arrive. Although sellers are relatively static in these online double auctions, they are not as same as online one-sided auctions, even those considering reserve prices, because not only buyers but also sellers are playing strategically in double auctions. 
Moreover, we will show in the conclusion how this assumption can be relaxed.

\section{A Deterministic \& Competitive Online Double Auction}
\label{sect_case1}
Although we just showed that in general deterministic mechanism is not competitive, in this section we demonstrate that a simple deterministic mechanism, called $\mathcal{M}_{greedy}$,  is actually $2$-competitive, given that the demand (i.e. the number of buyers) is not more than the supply (i.e. the number of sellers).

\subsection{Specification of $\mathcal{M}_{greedy}$}
The allocation policy, called {\it Best-first (Bf) Allocation}, of the deterministic ODA $\mathcal{M}_{greedy}$ greedily matches a newly arrived bid to the best unmatched ask, if they are matchable, until there is no unmatched ask left or all bids have arrived.
\begin{framed}
\noindent\textbf{The Allocation Policy of $\mathcal{M}_{greedy}$}\\
\rule{\textwidth}{0.5pt}
\begin{itemize}
\item Rank all asks $\theta^A$ in ascending order of their valuations (breaking ties randomly).
\end{itemize}
Upon arrival of bid $\theta_i^B$:
\begin{itemize}
\item If the unmatched ask $\theta_j^A$ with the highest ranking position is matchable with $\theta_i^B$, match $\theta_i^B$ with $\theta_j^A$. Otherwise, $\theta_i^B$ is unmatched.
\end{itemize}
\end{framed}

Figure~\ref{fig:bf1} shows an example of the greedy allocation, where dots indicate asks and bids, the value beside each dot represents the valuation of the ask/bid, and the order of the bids is their arrival order (from top to bottom). There is a line between an ask and a bid if they are matched by the allocation. Before we describe the payment policy of $\mathcal{M}_{greedy}$, let us first introduce a notion of reachability used in the payment policy. 

Let $((\theta_1^{A^*},\theta_1^{B^*}),(\theta_2^{A^*},\theta_2^{B^*}),...)$ be the sequence of ask-bid pairs that are matched by the greedy allocation in bid's arrival order, e.g. $((2,7),(3,4),(5,6))$ in the example shown in Figure~\ref{fig:best-first}, we say that two matched pairs $(\theta_i^{A^*},\theta_i^{B^*})$ and $(\theta_j^{A^*}, \theta_j^{B^*})$, where $i\leq j$, are \textbf{reachable} from each other, if for all $i\leq m< j$, bid $\theta_m^{B^*}$ and ask $\theta_{m+1}^{A^*}$ are matchable. For the example shown in Figure~\ref{fig:best-first}, $(2,7)$ and $(3,4)$ are reachable from each other, but $(5,6)$ is not reachable from $(2,7)$ and $(3,4)$ because ask of valuation 5 and bid of valuation 4 are not matchable.

\begin{figure}[ht]
   \centering
   \begin{minipage}[b]{.25\linewidth}
	\centering
	\includegraphics[height=0.135\textheight]{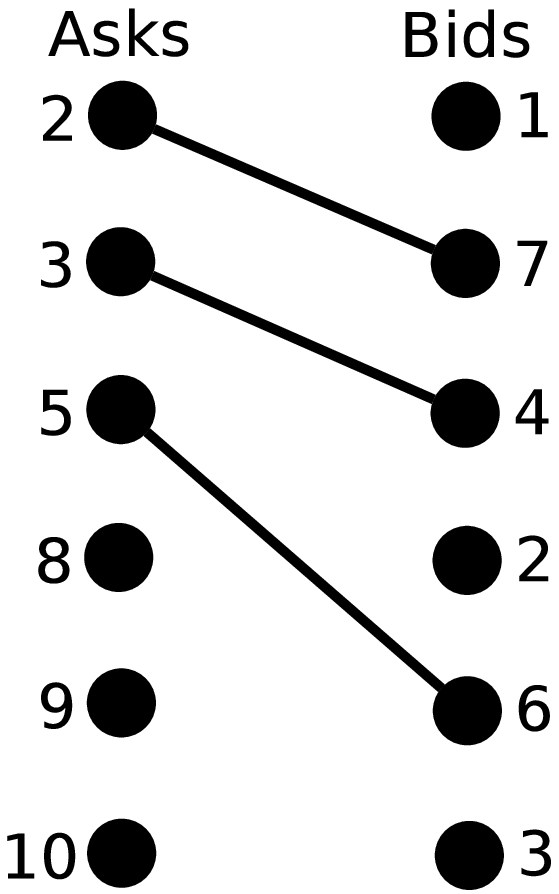}
	\subcaption{Best-first}\label{fig:bf1}
   \end{minipage}
   \begin{minipage}[b]{.25\linewidth}
   \centering
   \includegraphics[height=0.135\textheight]{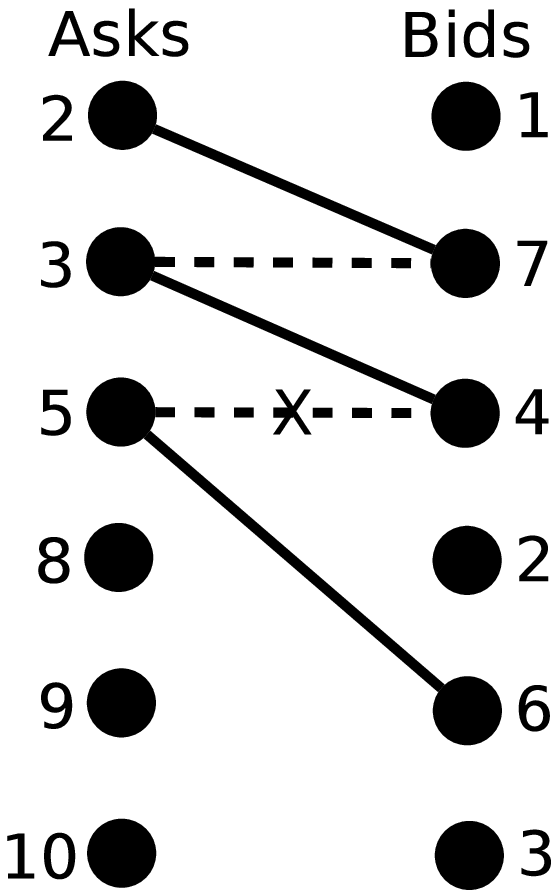}
   \subcaption{Reachability}\label{fig:bf_r}
   \end{minipage}
   \begin{minipage}[b]{.25\linewidth}
   \centering
   \includegraphics[height=0.135\textheight]{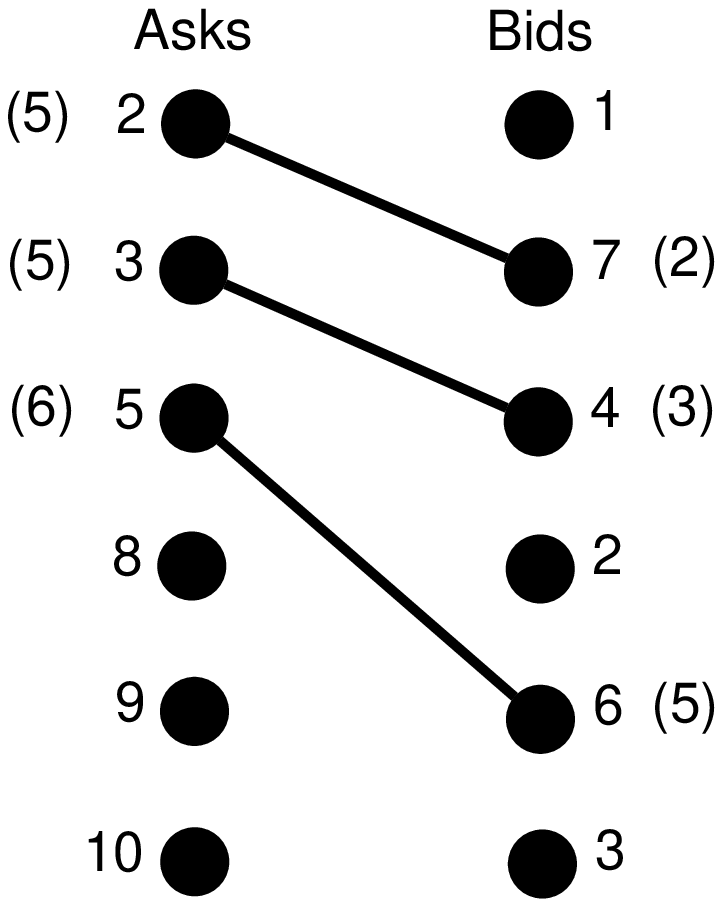}
   \subcaption{Payments}\label{fig:bf1_pay}
   \end{minipage}
   \caption{A Running Example of $\mathcal{M}_{greedy}$}
   \label{fig:best-first}
   \vspace{-2ex}
\end{figure}

The payment policy is described in the following, which shows a way to calculate the VCG payment (aka \textit{critical value}~\cite{Parkes_OnlineMD_2007}). Each matched buyer pays the amount equal to the valuation of the seller to whom he is matched, which is the infimum of all possible reported valuations for him to be matched in the auction, while each matched seller receives the supremum of all payments she can ask to get matched. There is no payment for unmatched traders, i.e. the mechanism is individually rational.
\begin{framed}
\noindent\textbf{The Payment Policy of $\mathcal{M}_{greedy}$}\\
\rule{\textwidth}{0.5pt}
\textbf{For each matched seller $i$ with type $\theta_i$:}
\begin{equation*}
x_i(\theta) = \left\{
   \begin{array}{ll}
    \min(v(\bar{\theta}_{min}^A),\max(v(\theta_{last}^B),v(\bar{\theta}_{max}^B))), & \text{if $\theta_{last}^B$ is reachable from $\theta_i$}\\
    \max(v(\theta_{last}^A),v(\bar{\theta}_{max}^B)), & \text{otherwise}\\
   \end{array}
 \right.
\end{equation*}
where 
\begin{itemize}
\item $\theta_{last}^A$ is the last matched ask, and $\theta_{last}^B$ is the last matched bid,
\item $\bar{\theta}_{min}^A$ is the unmatched ask with the lowest valuation ($v(\bar{\theta}_{min}^A) = \infty$ if $\bar{\theta}_{min}^A$ does not exist),
\item $\bar{\theta}_{max}^B$ is the unmatched bid with the highest valuation ($v(\bar{\theta}_{max}^B)= 0$ if $\bar{\theta}_{max}^B$ does not exist).
\end{itemize}
\textbf{For each matched buyer $j$ with type $\theta_j$:}
\begin{equation*}
x_j(\theta) = v(m(\theta_j)), \text{ where $m(\theta_j)$ is the ask matched to $\theta_j$.} 
\end{equation*}
\end{framed}

Example in Figure~\ref{fig:bf1_pay} shows the payments beside matched asks and bids according to the above payment rule. In this example, $v(\bar{\theta}_{min}^A)$ is 8, $v(\theta_{last}^A)$ is 5, $v(\bar{\theta}_{max}^B)$ is 3 and $v(\theta_{last}^B)$ is 6. It is easy to see that $\mathcal{M}_{greedy}$ is running a deficit in this example. In other words, $\mathcal{M}_{greedy}$ is not budget balanced, which is another important criterion that we cannot achieve at the same time in this work. 

\subsection{Properties of $\mathcal{M}_{greedy}$}
In the following, we prove that deterministic auction $\mathcal{M}_{greedy}$ is truthful and $2$-competitive.

\begin{theorem}
\label{the:truth}
 $\mathcal{M}_{greedy}$ is truthful. 
\end{theorem}

Instead of proving that the allocation is monotonic and that the payment is a kind of critical value~\cite{Parkes_OnlineMD_2007}, we demonstrate it in a more intuitive manner and the proof is given in the Appendix.

To check the efficiency of $\mathcal{M}_{greedy}$, we will apply competitive analysis, a method invented for analysing online algorithms. In other words, we will determine a competitive ratio $c$, defined in Definition~\ref{def_competitive}, for $\mathcal{M}_{greedy}$.

To that end, given a report profile $\theta$, we need to first know what is the optimal allocation, i.e. an allocation maximising social welfare, if we are aware of all future inputs/reports in advance. In this case, the optimal allocation is achieved by matching the highest bid (with respect to valuation) with the lowest ask, the second highest bid with the second lowest ask and so on, until there is no more matchable pair left. 
It is easy to check  that all asks that are matched by the optimal allocation are also matched by Best-first Allocation.


\begin{lemma}
\label{lem:matching}
All asks that are matched by the optimal allocation are also matched by Best-first Allocation.
\end{lemma}

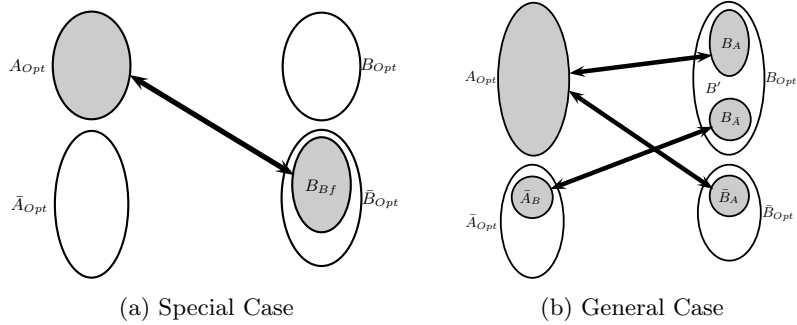
\begin{figure}[ht]
\vspace{-2ex}
   \centering
   \begin{minipage}[b]{.45\linewidth}
	\centering
\scalebox{0.7} 
{
\begin{pspicture}(0,-2.54)(7.0903125,2.54)
\definecolor{color454b}{rgb}{0.8,0.8,0.8}
\psellipse[linewidth=0.04,dimen=outer,fillstyle=solid,fillcolor=color454b](1.3871875,1.5)(0.75,1.04)
\psellipse[linewidth=0.04,dimen=outer](1.3871875,-1.13)(0.71,1.41)
\psellipse[linewidth=0.04,dimen=outer](5.7071877,-1.01)(0.77,1.31)
\psellipse[linewidth=0.04,dimen=outer,fillstyle=solid,fillcolor=color454b](5.7071877,-0.76)(0.57,0.92)
\psline[linewidth=0.1cm,arrowsize=0.05291667cm 2.0,arrowlength=1.4,arrowinset=0.4]{<->}(2.0971875,1.32)(5.1971874,-0.58)
\usefont{T1}{ptm}{m}{n}
\rput(0.18375,1.47){$A_{Opt}$}
\usefont{T1}{ptm}{m}{n}
\rput(0.21828125,-1.15){$\bar{A}_{Opt}$}
\usefont{T1}{ptm}{m}{n}
\rput(6.7821875,1.49){$B_{Opt}$}
\usefont{T1}{ptm}{m}{n}
\rput(6.8167186,-1.03){$\bar{B}_{Opt}$}
\usefont{T1}{ptm}{m}{n}
\rput(5.7026563,-0.81){$B_{Bf}$}
\psellipse[linewidth=0.04,dimen=outer](5.7071877,1.48)(0.75,1.04)
\end{pspicture} 
}
	\subcaption{Special Case}\label{fig:the_2com}
   \end{minipage}
   \begin{minipage}[b]{.45\linewidth}
	\centering
\scalebox{0.6} 
{
\begin{pspicture}(0,-3.07)(6.993125,3.07)
\definecolor{color91b}{rgb}{0.8,0.8,0.8}
\psellipse[linewidth=0.04,dimen=outer,fillstyle=solid,fillcolor=color91b](1.3371875,1.34)(0.8,1.71)
\psellipse[linewidth=0.04,dimen=outer](5.6171875,1.36)(0.8,1.71)
\psellipse[linewidth=0.04,dimen=outer](1.3071876,-1.8)(0.71,1.27)
\psellipse[linewidth=0.04,dimen=outer](5.6271877,-1.61)(0.71,1.08)
\psellipse[linewidth=0.04,dimen=outer,fillstyle=solid,fillcolor=color91b](5.6271877,2.14)(0.45,0.75)
\psellipse[linewidth=0.04,dimen=outer,fillstyle=solid,fillcolor=color91b](5.6471877,0.45)(0.47,0.48)
\psellipse[linewidth=0.04,dimen=outer,fillstyle=solid,fillcolor=color91b](1.3071876,-1.27)(0.47,0.48)
\psellipse[linewidth=0.04,dimen=outer,fillstyle=solid,fillcolor=color91b](5.6271877,-1.24)(0.45,0.47)
\psline[linewidth=0.1cm,arrowsize=0.05291667cm 2.0,arrowlength=1.4,arrowinset=0.4]{<->}(2.1171875,1.47)(5.2571874,1.87)
\psline[linewidth=0.1cm,arrowsize=0.05291667cm 2.0,arrowlength=1.4,arrowinset=0.4]{<->}(1.6971875,-1.09)(5.2771873,0.29)
\psline[linewidth=0.1cm,arrowsize=0.05291667cm 2.0,arrowlength=1.4,arrowinset=0.4]{<->}(2.0971875,1.09)(5.2771873,-1.09)
\usefont{T1}{ptm}{m}{n}
\rput(0.18375,1.36){$A_{Opt}$}
\usefont{T1}{ptm}{m}{n}
\rput(0.19828124,-1.84){$\bar{A}_{Opt}$}
\usefont{T1}{ptm}{m}{n}
\rput(6.7621875,1.34){$B_{Opt}$}
\usefont{T1}{ptm}{m}{n}
\rput(6.6967187,-1.62){$\bar{B}_{Opt}$}
\usefont{T1}{ptm}{m}{n}
\rput(5.6692185,2.12){$B_A$}
\usefont{T1}{ptm}{m}{n}
\rput(5.6579685,0.46){$B_{\bar{A}}$}
\usefont{T1}{ptm}{m}{n}
\rput(5.2723436,1.16){$B^\prime$}
\usefont{T1}{ptm}{m}{n}
\rput(5.617344,-1.26){$\bar{B}_A$}
\usefont{T1}{ptm}{m}{n}
\rput(1.2707813,-1.28){$\bar{A}_B$}
\end{pspicture}  
}
\subcaption{General Case}\label{fig:the_2com2}
   \end{minipage}
\caption{Best-first Allocation of $\mathcal{M}_{greedy}$}
\vspace{-2ex}
\end{figure}

\begin{theorem}
\label{the_2com}
 $\mathcal{M}_{greedy}$ is $2$-competitive.
\end{theorem}
\begin{proof}
We first show that this competitive ratio is achievable under a special case, and then we prove that in any other cases the ratio is also achievable.

The special case is that all matched asks of the optimal allocation are matched to unmatched bids of the optimal allocation by Best-first Allocation, and unmatched asks of the optimal allocation are also not matched by Best-first Allocation (see Figure~\ref{fig:the_2com} for example, where the coloured areas are the asks and bids matched by Best-first Allocation and double-sided arrows indicate the matching relation). Let $A_{Opt}$ and $\bar{A}_{Opt}$ be the matched and unmatched asks respectively in the optimal allocation, and $B_{Opt}$ and $B_{Bf}$ be the matched bids in the optimal allocation and Best-first Allocation respectively and $\bar{B}_{Opt}$ and $\bar{B}_{Bf}$ be the corresponding unmatched bids. We can get that $B_{Opt} \cap B_{Bf} = \emptyset$, i.e. no bid from $B_{Opt}$ can be matched to any ask from $\bar{A}_{Opt}$. We also know that $\bar{A}_{Opt} \neq \emptyset$ and $|\bar{A}_{Opt}| \geq |B_{Opt}|$ because we assumed that the demand is not more than the supply. Therefore,
\begin{equation}
\label{equ1}
\sum_{\theta_i\in \bar{A}_{Opt}}v(\theta_i) > \sum_{\theta_i\in B_{Opt}}v(\theta_i). 
\end{equation}
The social welfare of the optimal allocation is:
\begin{equation}
\label{equ2}
W(Opt(\theta)) = \sum_{\theta_i\in \bar{A}_{Opt}}v(\theta_i) + \sum_{\theta_i\in B_{Opt}}v(\theta_i).
\end{equation}
The social welfare of Best-first Allocation is:
\begin{equation}
\label{equ3}
W(Bf(\theta)) = \sum_{\theta_i\in \bar{A}_{Opt}}v(\theta_i) + \sum_{\theta_i\in B_{Bf}}v(\theta_i).
\end{equation}
Combining \eqref{equ1}, \eqref{equ2} and \eqref{equ3}, we get
\begin{equation*}
\label{equ4}
\frac{W(Bf(\theta))}{W(Opt(\theta))} >
\frac{\sum_{\theta_i\in \bar{A}_{Opt}}v(\theta_i) + \sum_{\theta_i\in B_{Bf}}v(\theta_i)}{\sum_{\theta_i\in \bar{A}_{Opt}}v(\theta_i) + \sum_{\theta_i\in \bar{A}_{Opt}}v(\theta_i)} >
\frac{1}{2}.
\end{equation*}

So far, we have proved the theorem in a special case. In general case, some asks of $A_{Opt}$ might be matched to some bids of $B_{Opt}$, and some asks of $\bar{A}_{Opt}$ might be matched to some bids of $B_{Opt}$ by Best-first Allocation. Due to Lemma~\ref{lem:matching}, we know that all asks in $A_{Opt}$ are matched by Best-first Allocation. Let $B_A$ and $\bar{B}_A$ be all the bids from $B_{Opt}$ and $\bar{B}_{Opt}$ respectively that are matched to asks of $A_{Opt}$ by Best-first Allocation. Let $\bar{A}_B$ be the asks from $\bar{A}_{Opt}$ that are matched to some bids of $B_{Opt}$ by Best-first Allocation, and $B_{\bar{A}}$ be the corresponding bids matched to $\bar{A}_{Opt}$. Let $B^\prime = B_{Opt} \setminus (B_A \cup B_{\bar{A}})$ be the asks from $B_{Opt}$ that are not matched by Best-first Allocation (see Figure~\ref{fig:the_2com2}). Therefore, the social welfare of Best-first Allocation is:
\begin{equation*}
 W(Bf(\theta)) = \sum_{\theta_i\in \bar{A}_{Opt}\setminus\bar{A}_B}v(\theta_i) + \sum_{\theta_i\in B_A \cup \bar{B}_A \cup B_{\bar{A}}}v(\theta_i).
\end{equation*}
So, we get
\begin{align}
\label{equ5}
 \frac{W(Bf(\theta))}{W(Opt(\theta))} 
&= \frac{\sum_{\theta_i\in \bar{A}_{Opt}\setminus\bar{A}_B}v(\theta_i) + \sum_{\theta_i\in B_A \cup \bar{B}_A \cup B_{\bar{A}}}v(\theta_i)}{\sum_{\theta_i\in \bar{A}_{Opt}}v(\theta_i) + \sum_{\theta_i\in B_{Opt}}v(\theta_i)} \nonumber \\
&= \frac{\sum_{\theta_i\in \bar{A}_{Opt} \cup B_{Opt}}v(\theta_i) - \varSigma}{\sum_{\theta_i\in \bar{A}_{Opt} \cup B_{Opt}}v(\theta_i)} \nonumber \\
&= 1 - \frac{\varSigma}{\sum_{\theta_i\in \bar{A}_{Opt} \cup B_{Opt}}v(\theta_i)},
\end{align}
where $\varSigma = \sum_{\theta_i\in B^\prime \cup \bar{A}_B}v(\theta_i) - \sum_{\theta_i\in \bar{B}_A}v(\theta_i)$.

Since the number of bids is not more than that of asks, i.e. the number of unmatched bids is not more than that of unmatched asks, we get $|\bar{A}_{Opt}\setminus \bar{A}_B| > |B^\prime|$. We know that no ask from $\bar{A}_{Opt}\setminus \bar{A}_B$ can be matched to any bid in $B^\prime$, so $\sum_{\theta_i \in \bar{A}_{Opt}\setminus \bar{A}_B} v(\theta_i) \geq \sum_{\theta_i \in B^\prime} v(\theta_i)$, i.e. $\sum_{\theta_i \in \bar{A}_{Opt}} v(\theta_i) \geq \sum_{\theta_i \in B^\prime \cup \bar{A}_B} v(\theta_i)$. Thus,
\begin{equation}
\label{equ6}
 \frac{\varSigma}{\sum_{\theta_i\in \bar{A}_{Opt} \cup B_{Opt}}v(\theta_i)} \leq 
 \frac{\varSigma}{\sum_{\theta_i\in B^\prime \cup \bar{A}_B}v(\theta_i) + \sum_{\theta_i\in B_{Opt}}v(\theta_i)}.
\end{equation}
Since
\begin{equation*}
 \sum_{\theta_i\in B_{Opt}} v(\theta_i) = 
 \sum_{\theta_i\in B_{\bar{A}} \cup B_{A} \cup B^\prime} v(\theta_i) \geq 
 \sum_{\theta_i\in \bar{A}_B \cup B_{A} \cup B^\prime} v(\theta_i) \geq 
 \sum_{\theta_i\in \bar{A}_B \cup B^\prime} v(\theta_i),
\end{equation*}
we conclude that 
\begin{align}
\label{equ7}
& \frac{\varSigma}{\sum_{\theta_i\in B^\prime \cup \bar{A}_B}v(\theta_i) + \sum_{\theta_i\in B_{Opt}}v(\theta_i)} \leq \nonumber \\
& \frac{\varSigma}{\sum_{\theta_i\in B^\prime \cup \bar{A}_B}v(\theta_i) + \sum_{\theta_i\in \bar{A}_B \cup B^\prime}v(\theta_i)} \leq \nonumber \\
& \frac{\varSigma + \sum_{\theta_i\in \bar{B}_A}v(\theta_i)}{\sum_{\theta_i\in B^\prime \cup \bar{A}_B}v(\theta_i) + \sum_{\theta_i\in \bar{A}_B \cup B^\prime}v(\theta_i)} = \frac{1}{2}.
\end{align}
Combining \eqref{equ5}, \eqref{equ6} and \eqref{equ7}, we get $\frac{W(Bf(\theta))}{W(Opt(\theta))} \geq \frac{1}{2}$. \qed
\end{proof}

\section{Reducing Double Auctions to One-sided Auctions}
\label{sect_case2}
In this section, we study another case where we can predict how many buyers will arrive. Given this prior information, we propose a reduction framework which reduces an ODA to an \textit{online one-sided auction} that aims to select the $k$-best bids from $n$ bids arriving in an online fashion, e.g. secretary-problem-based online auctions~\cite{Hajiaghayi_2004,Kleinberg_2005,Buchbinder_2010}.

The main difference between ODAs and online one-sided auctions is that, instead of allocating $k$ items to $n$ agents in a one-sided auction, we do not know how many items we should allocate to buyers in ODAs, because items are provided by sellers which are unpredictable. Moreover, it is not efficient to allocate an item from a seller with a high valuation to a buyer with a lower valuation. For instance, in a double auction with only one seller, the auctioneer does not just select any buyer but the one with a valuation at least better than the seller's. Since the goal of an efficient double auction is to allocate items to traders with higher valuations, we can actually treat sellers as additional buyers and apply efficient one-sided auction. In the rest of this section, we will show how to consider sellers as additional buyers to design truthful and competitive ODAs by applying truthful and competitive online one-sided auctions.

\subsection{The Reduction}
Let $n^A$ and $n^B$ be the number of asks $\theta^A$ and bids $\theta^B$ respectively. Let $\mathcal{A}$ be an online one-sided auction. We construct an ODA $\mathcal{M_A}$ from $\mathcal{A}$ as follows. The intuition is considering sellers as additional buyers by giving asks opportunity to compete with bids in order to gain items back for sellers, if sellers' valuations are comparatively high among the valuations of both sellers and buyers. By doing this, a seller with a comparatively high valuation will have a comparatively high chance to get her item back if maximising social welfare is an objective of $\mathcal{A}$. In order to treat relatively static sellers as buyers, we assign them a new online arrival order which is consistent with the arrival of buyers.

\begin{framed}
\noindent\textbf{Online Double Auction $\mathcal{M_A}$ based on Online One-sided Auction $\mathcal{A}$}\\
\rule{\textwidth}{0.5pt}
 \begin{enumerate}
  \item Choose a position $l_i \in [1,n^A+n^B]$ for each ask $\theta_i$ according to a discrete probability distribution function $f(x)$ that satisfies the assumptions made on the arrival order of buyers.
  \item Run $\mathcal{A}$ on the inputs that contain both asks $\theta^A$ and bids $\theta^B$ where each ask $\theta_i$ arrives right after the $(l_i-1)$-th input arrived.
  \item If a bid $\theta_i$ is selected by $\mathcal{A}$ with payment $p_i$ and $v(\theta_i) \geq v(\theta_j)$, where $\theta_j$ is the currently unmatched ask with lowest valuation (breaking ties randomly), then $\theta_i$ is matched to $\theta_j$ with payment 
\begin{equation}
\label{eq_payBid}
 x_i(\theta) = \max(p_i,v(\theta_j)).
\end{equation}
Otherwise, $\theta_i$ is unmatched.
  \item Once the matching/allocation is done, the payment for each matched ask $\theta_j$ is as same as the one defined in $\mathcal{M}_{greedy}$, except that the bids considered in the payment here are those selected by $\mathcal{A}$ only.
\end{enumerate}
\end{framed}

For the probability distribution function $f(x)$ of $\mathcal{M_A}$, we only require that $f(x)$ satisfies the assumptions made on the arrival order of the inputs of $\mathcal{A}$. In other words, the arrival order assigned to asks satisfies the assumptions made on the arrival order of bids. For instance, if $\mathcal{A}$ is based on a random-ordering model, e.g. secretary-problem-based online auctions~\cite{Kleinberg_2005}, then $f(x)$ can only be a random distribution function. If $\mathcal{A}$ is based on an adversary-ordering model, then $f(x)$ can be any distribution function. More interestingly, if $\mathcal{A}$ has no assumption made on the arrival order of its inputs, we can utilise $f(x)$ for other purpose. In single-seller case, for example, we might push the ask to the front of the inputs to guarantee a higher expected valuation of the selected trader and therefore to further improve the efficiency of $\mathcal{M_A}$.

Figure~\ref{fig:MA_example} shows a running example of $\mathcal{M_A}$. $\mathcal{M_A}$ first chooses a position for each ask, then runs $\mathcal{A}$ on the merged input and selects the winners (indicated by `*'), and finally determines the final asks and bids that are matched by using the winners selected by $\mathcal{A}$ (traders allocated an item by $\mathcal{M_A}$ are indicated by circles). From the example in Figure~\ref{fig:MA_example}, we can say that both the ask of value 2 and the bid of value 6 do not get item in the end, although they are selected by $\mathcal{A}$. That is, $\mathcal{M_A}$ might improve the social welfare of the allocation given by $\mathcal{A}$.

\begin{figure}[ht]
   \centering
  \includegraphics[width=0.45\textwidth]{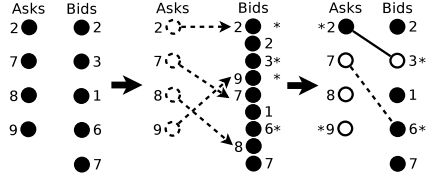}
   \caption{A Running Example of $\mathcal{M_A}$}
  \label{fig:MA_example}
   \vspace{-2ex}
\end{figure}

\subsection{Key Properties of $\mathcal{M_A}$}
We will prove that the truthfulness and efficiency of $\mathcal{M_A}$ directly follow that of the one-sided auction $\mathcal{A}$, and then show two instances of $\mathcal{M_A}$ by utilizing secretary-based online auctions.

\begin{theorem}
 If $\mathcal{A}$ is truthful, then $\mathcal{M_A}$ is truthful.
\end{theorem}

\begin{proof}
We will prove for sellers and buyers respectively. We need to show that both sellers and buyers will reveal their true valuation, arrive and departure truthfully, i.e. traders are incentivized to arrive as early as they can and depart as late as possible.

For a buyer $i$ of type $\theta_i$ that is not selected by $\mathcal{A}$, $\theta_i$ will also not be matched by $\mathcal{M_A}$. If $i$ misreported $\theta_i^\prime$ and is selected by $\mathcal{A}$, then $v(\theta_i) - p_i \leq 0$, i.e. $i$ will get a negative (expected) utility in $\mathcal{A}$, because $\mathcal{A}$ is truthful. Therefore, if $\theta_i^\prime$ is matched by $\mathcal{M_A}$, then $i$'s utility $v(\theta_i) - \max(p_i,v(\theta_j)) \leq v(\theta_i) - p_i \leq 0$. 

For a buyer $i$ of type $\theta_i$ that is selected by $\mathcal{A}$, $\theta_i$ will be either matched or unmatched by $\mathcal{M_A}$ depending on $v(\theta_i)$ and the lowest unmatched ask $\theta_j$ when $\theta_i$ is selected. If $v(\theta_i) \geq v(\theta_j)$, $\theta_i$ is matched by $\mathcal{M_A}$. Otherwise, $\theta_i$ is unmatched. If $\theta_i$ is matched by $\mathcal{M_A}$, then we know that $i$'s utility $v(\theta_i) - \max(p_i,v(\theta_j))$ is maximised, because $v(\theta_i) - p_i$ is maximised by $\mathcal{A}$ and $v(\theta_j)$ is independent of $i$ and it is minimised if $i$ arrives at his earliest arrival time. If $\theta_i$ is not matched by $\mathcal{M_A}$, then we have $p_i \leq v(\theta_i) < v(\theta_j)$. Since $v(\theta_j)$ is independent of $i$ and it is minimised if $i$ arrives at his earliest arrival time, $i$ can only be matched by $\mathcal{M_A}$ if $i$ misreported $\theta_i^\prime$ such that $v(\theta_i^\prime) \geq v(\theta_j)$, but then his utility $v(\theta_i) - \max(p_i,v(\theta_j)) < 0$.

We conclude from the above that buyers are incentivized to arrive at their earliest arrival time and report their true valuation. Moreover, $\mathcal{M_A}$ does not use their departure time for decision-making, so the truthfulness of their departure directly follows that of $\mathcal{A}$.


For sellers, since we assume that all sellers are patient, i.e. sellers arriving after the arrival of the first bid or departing before the last bid's arrival are not considered by $\mathcal{M_A}$, all sellers are incentivized to arrive/depart truthfully. 
The following proves that sellers are also incentivized to reveal their true valuation. 

For a matched seller $i$ with ask $\theta_i$, from the truthfulness of $\mathcal{M}_{greedy}$, we know that the payment of $\mathcal{M_A}$ also guarantees truthfulness for sellers, if the bids selected by $\mathcal{A}$ are the same when $i$ reported differently. 
However, the bids selected by $\mathcal{A}$ might change if $i$ reported a different valuation, so we need to check that the changes are not beneficial for $i$. Since $\mathcal{A}$ is truthful, the wining probability for trader $i$ with valuation report $v_i^\prime > v_i$ should be at least that with valuation report $v_i$ in $\mathcal{A}$ (aka monotonicity). If $i$ reported a higher valuation, she might lower the winning probabilities of others/buyers. That is, the winning probability of $i$ with a higher valuation report in $\mathcal{M_A}$ might be decreased and also the payment will be potentially decreased. Thus, it is not beneficial for $i$ to misreport a higher valuation. If $i$ misreported a lower valuation, then the winning probability of others might be increased, and therefore, the bids selected by $\mathcal{A}$ will have relatively lower valuations and more bids might be selected. Since the number of agents $\mathcal{A}$ can select is fixed, by misreporting a lower valuation, $i$ increases the chance for lower-value buyers to get matched and receives a lower pay. Thus, $i$ reduces her chance to get matched with positive utility by reporting a lower valuation. Similarly, we can check for unmatched sellers.
%
\qed
\end{proof}

\begin{theorem}
 If $\mathcal{A}$ is $c$-competitive, then $\mathcal{M_A}$ is $c$-competitive.
\end{theorem}
\begin{proof}
Given report profile $\theta$, let $A_{\mathcal{A}}$ and $B_{\mathcal{A}}$ be the sets of selected asks and bids in $\mathcal{A}$ respectively. Since $\mathcal{A}$ is $c$-competitive for maximising social welfare, we get 
 $W(\mathcal{A}(\theta)) = \sum_{\theta_i \in A_{\mathcal{A}} \cup B_{\mathcal{A}}} v(\theta_i) \geq \frac{W(Opt(\theta))}{c}$.
Based on the winners $A_{\mathcal{A}} \cup B_{\mathcal{A}}$ selected by $\mathcal{A}$, $\mathcal{M_A}$ will further improve the allocation. More specifically, an ask selected (unselected) by $\mathcal{A}$ might sell (hold) the item in $\mathcal{M_A}$ (e.g. the ask of value 2 in Figure~\ref{fig:MA_example}), while a bid selected by $\mathcal{A}$ might not be matched by $\mathcal{M_A}$ if the bid's valuation is comparatively lower (e.g. the bid of value 6 in Figure~\ref{fig:MA_example}). The reason is that $A_{\mathcal{A}}$ is only used to determine at least the $|A_{\mathcal{A}}|$-best sellers will keep their items, and that some bids of $B_{\mathcal{A}}$ are not matched by $\mathcal{M_A}$ if their valuations are not good enough. 
Thus,
 $W(\mathcal{M_A}(\theta)) \geq W(\mathcal{A}(\theta)) \geq \frac{W(Opt(\theta))}{c}$.
\qed
\end{proof}

\begin{corollary}
\label{cor1}
 Let $k$ be the number of sellers, there exists a truthful ODA $\mathcal{M_A}$ that is 
\begin{itemize}
 \item $2\sqrt{e}$-competitive for $k=1$.
 \item $(1+\frac{C}{\sqrt{k}})$-competitive.
\end{itemize}
\end{corollary}

Corollary~\ref{cor1} follows the $2\sqrt{e}$-competitive online single-item auction proposed by Buchbinder \textit{et al.}~\cite{Buchbinder_2010} via linear programming and the $(1+\frac{C}{\sqrt{k}})$-competitive online multi-item auction introduced by Kleinberg~\cite{Kleinberg_2005}, which approaches to $1$-competitive as $k$ approaches to $\infty$. These two online one-sided auctions are based on secretary problems, i.e. traders arrive randomly and therefore $f(x)$ of $\mathcal{M_A}$ is an uniform random distribution function in these instances.

It is worth mentioning that the reduction approach is also applicable if we do not know how many buyers will arrive but that their arrival time satisfies some distribution. In that case, we will assign an arrival time for each seller in the reduction following that distribution.

\section{Conclusion}
\label{sect_con}
We have studied the mechanism design problem of online double auction markets where traders are dynamically arriving and departing the markets. Due to the complexity of the dynamics brought by traders, 
we showed that there is no deterministic and truthful online double auction that is competitive for maximising social welfare in an adversarial model. However, this impossibility does not apply to the situations where we can access certain prior information of the participants. In this paper, we studied two environments where sellers are relatively static and certain prior information of buyers is accessible. In the first environment, we assumed that the demand (i.e. the number of buyers) is not more than the supply (i.e. the number of sellers). Under this assumption, we proposed a deterministic yet $2$-competitive and truthful online mechanism in Section~\ref{sect_case1}. In the second environment, given the prior information that the number of incoming buyers is predictable, we demonstrated in Section~\ref{sect_case2} how to reduce a truthful online double auction to a truthful online one-sided auction, and showed that the competitiveness of the reduced online double auction follows that of the online one-sided auction. Especially, by using the reduction framework, 
we found an online double auction that is almost $1$-competitive. However, the mechanisms proposed in this paper are not (weakly) budget balanced, which is also an important factor besides truthfulness and efficiency and worth further investigation, though it is often very hard to achieve all three criteria together even in static cases~\cite{Myerson_1983,Gonen_2007}.
In addition, there are many other online exchanges are worth further investigation, e.g. electric vehicle charging \cite{GerdingSRZJ13} and kidney exchange~\cite{Utku_2010}.

One might suspect that the ``static" assumption made on sellers will limit the applicability of these mechanisms. We argue that they can be applied in more general settings where sellers can also arrive and depart at anytime. One way to apply these mechanisms is running multiple instances of them in sequence. In other words, we decompose an online market into multiple disjoint sub-markets where the conditions fit the assumptions made here. For example, in some exchange markets, both sellers and buyers come and leave randomly, but one side, e.g. sellers, stay longer than the other side. In that case, we can decompose the market into many sub-markets running for a period of, say, one-month, i.e. there will be 12 disjoint sub-markets for a one-year market. Each trader is allocated to one/many sub-markets on his/her arrival, and the decomposition should guarantee that each seller is able to fully participate in at least one sub-market. Under this decomposition, if the market is in a rising situation, then applying the proposed auctions in each sub-market will achieve the same truthfulness and competitiveness for the whole market (see the Appendix for more details). 

We have seen that different prior knowledge gives us different advantages for designing online mechanisms, as it reduces the dynamics in some extent. 
Especially, in very complex dynamic environments, without certain prior knowledge, in general it is impossible to get ideal mechanisms. 
Therefore, one objective of mechanism design in such complex environments is searching for desirable mechanisms by utilising as less prior knowledge as possible. 
Besides prior knowledge, randomisation has also played an important role in this paper and other online algorithm design~\cite{Karp:1990,Ben-David_1990,Chrobak:2008}.



\begin{thebibliography}{10}

\bibitem{Myerson_1983}
Myerson, R.B., Satterthwaite, M.A.:
\newblock Efficient mechanisms for bilateral trading.
\newblock Journal of Economic Theory \textbf{29}(2) (April 1983)  265--281

\bibitem{McAfee_1990}
McAfee, R.P.:
\newblock A dominant strategy double auction.
\newblock Journal of Economic Theory \textbf{56}(2) (1992)  434--450

\bibitem{Wurman_1998}
Wurman, P.R., Walsh, W.E., Wellman, M.P.:
\newblock Flexible double auctions for electionic commerce: theory and
  implementation.
\newblock Decis. Support Syst. \textbf{24} (1998)  17--27

\bibitem{blum_online_2006}
Blum, A., Sandholm, T., Zinkevich, M.:
\newblock Online algorithms for market clearing.
\newblock J. {ACM} \textbf{53}(5) (2006)  845--879

\bibitem{Bredin_2007}
Bredin, J., Parkes, D.C., Duong, Q.:
\newblock Chain: a dynamic double auction framework for matching patient
  agents.
\newblock J. Artif. Int. Res. \textbf{30}(1) (2007)  133--179

\bibitem{Parkes_OnlineMD_2007}
Parkes, D.C.:
\newblock Online mechanisms.
\newblock In: Algorithmic Game Theory.
\newblock Cambridge University Press (2007)

\bibitem{Mehta:2007}
Mehta, A., Saberi, A., Vazirani, U., Vazirani, V.:
\newblock Adwords and generalized online matching.
\newblock J. ACM \textbf{54} (2007)

\bibitem{Myerson:2008}
Myerson, R.B.:
\newblock revelation principle.
\newblock In Durlauf, S.N., Blume, L.E., eds.: The New Palgrave Dictionary of
  Economics.
\newblock Palgrave Macmillan (2008)

\bibitem{Hajiaghayi_2004}
Hajiaghayi, M.T., Kleinberg, R., Parkes, D.C.:
\newblock Adaptive limited-supply online auctions.
\newblock In: EC'04: Proceedings of the 5th ACM Conference on Electronic
  Commerce, ACM (2004)  71--80

\bibitem{Kleinberg_2005}
Kleinberg, R.:
\newblock A multiple-choice secretary algorithm with applications to online
  auctions.
\newblock In: SODA'05: Proceedings of the 16th Annual ACM-SIAM Symposium on
  Discrete Algorithms, Philadelphia, PA, USA (2005)  630--631

\bibitem{Buchbinder_2010}
Buchbinder, N., Jain, K., Singh, M.:
\newblock Incentives in online auctions via linear programming.
\newblock In: Proceedings of the 6th International Conference on Internet and
  Network Economics. WINE'10, Berlin, Heidelberg, Springer-Verlag (2010)
  106--117

\bibitem{Gonen_2007}
Gonen, M., Gonen, R., Pavlov, E.:
\newblock Generalized trade reduction mechanisms.
\newblock In: Proceedings of the 8th ACM Conference on Electronic Commerce.
  EC'07, New York, NY, USA, ACM (2007)  20--29

\bibitem{Karp:1990}
Karp, R.M., Vazirani, U.V., Vazirani, V.V.:
\newblock An optimal algorithm for on-line bipartite matching.
\newblock In: Proceedings of the 22nd Annual ACM Symposium on Theory of
  Computing. STOC'90, New York, NY, USA, ACM (1990)  352--358

\bibitem{Ben-David_1990}
Ben-David, S., Borodin, A., Karp, R., Tardos, G., Wigderson, A.:
\newblock On the power of randomization in online algorithms.
\newblock In: Algorithmica. (1990)  379--386

\bibitem{Chrobak:2008}
Chrobak, M.:
\newblock Sigact news online algorithms column 13: 2007 - an offine
  perspective.
\newblock SIGACT News \textbf{39} (2008)  96--121

\bibitem{GerdingSRZJ13}
Gerding, E.H., Stein, S., Robu, V., Zhao, D., Jennings, N.R.:
\newblock Two-sided online markets for electric vehicle charging.
\newblock In: AAMAS. (2013)  989--996

\bibitem{Utku_2010}
\"{U}nver, M.U.:
\newblock Dynamic kidney exchange.
\newblock Review of Economic Studies \textbf{77}(1) (2010)  372--414

\end{thebibliography}

\section*{Appendix:}

\subsection*{Proof of Theorem \ref{the:truth}}
\begin{proof}
We will prove the theorem for buyers and sellers respectively. 

\textit{\underline{For buyers:}}
Since the payment for matched buyers are non-decreasing over time because of the valuation increasing of the lowest unmatched ask, the earlier the arrival time a buyer has, the higher probability to be matched and the lower payment the buyer will get. Therefore, all buyers are incentivized to arrive at their true/earliest arrival time. Since the mechanism does not use buyer's departure time for decision-making, there is no motivation for buyers to misreport their departure time.

Regarding their valuation reporting, for a matched buyer $i$ with bid $\theta_i$, assume $m(\theta_i) = \theta_j$, i.e. $\theta_i$ is matched to $\theta_j$. $i$'s payment only depends on $v(\theta_j)$ and $v(\theta_j)$ is independent of $\theta_i$, so the payment of $i$ cannot be changed by $v(\theta_i)$. Moreover, increasing $v(\theta_i)$ does not change the probability for $\theta_i$ to be matched, while decreasing $v(\theta_i)$ will reduce the probability for $\theta_i$ to be matched. For an unmatched buyer $i$ with bid $\theta_i$, since $\theta_i$ cannot be matched to the currently best unmatched ask on the arrival of $\theta_i$ or there is no unmatched ask left, $i$ might be able to increase his valuation to get matched, but then he has to pay more than his valuation, i.e. $i$ gets negative utility. Thus, reporting valuation truthfully gives buyers the highest expected utility.

\textit{\underline{For sellers:}}
All sellers are incentivized to arrive and depart truthfully as they will not be considered if they arrive after the first buyer's arrival or depart before the last buyer's arrival.

For a matched seller $i$ with ask $\theta_i$, we will show that $i$ cannot report a different valuation other than her true valuation to improve her payment. Let $m$ be the matching given by Best-first Allocation and $\theta_j = m(\theta_i)$. Assume that $\theta_i$ and $\theta_j$ is the $i$-th matched pair in $m$ and $|m|=k$, i.e. the $k$-th matched pair of $m$ is $\theta_{last}^A$ and $\theta_{last}^B$. The following proof is given on the condition whether or not $\theta_{last}^B$ is reachable from $\theta_i$. 

\noindent(1) \textit{$\theta_{last}^B$ is reachable from $\theta_i$ (e.g. ask 2 in Figure~\ref{fig:bf4}):} 
\begin{itemize}
\item If $i$ reported $\theta_i^\prime$ instead of $\theta_i$ such that $v(\theta_i^\prime) > v(\theta_i)$ and \\$v(\theta_i^\prime) \leq \min(v(\bar{\theta}_{min}^A),\max(v(\theta_{last}^B),v(\bar{\theta}_{max}^B)))$, the ranking position of $\theta_i^\prime$ is $i^\prime \geq i$ and the allocation will give a new matching $m^\prime$. For all $i\leq l <i^\prime$, the $l$-th matched bid of $m$ will be matched to $(l+1)$-th matched ask of $m$ in $m^\prime$, $\theta_i^\prime$ will be matched to $i^\prime$-th matched bid of $m$ in $m^\prime$, and for all $1 \leq l <i$ and $i^\prime < l \leq k$, the $l$-th matched pair of $m$ is also a matched pair in $m^\prime$ (see Figure~\ref{fig:bf5} and \ref{fig:bf6} for example). In both $m$ and $m^\prime$, the payment for $i$ is the same because $\theta_{last}^B$ is still reachable from $\theta_i^\prime$, and $\bar{\theta}_{min}^A$, $\theta_{last}^B$ and $\bar{\theta}_{max}^B$ are not changed. Moreover, the probability for trader $i$ to be matched will be the same with both $\theta_i$ and $\theta_i^\prime$, which is $1$ here. However, if $v(\theta_i^\prime) > \min(v(\bar{\theta}_{min}^A), \max(v(\theta_{last}^B),v(\bar{\theta}_{max}^B)))$, then $\theta_i^\prime$ will not be matched in $m^\prime$ (see Figure~\ref{fig:bf9} for example). Therefore, $i$ cannot report a higher valuation to receive more payment.
\item If $i$ reported $\theta_i^\prime$ instead of $\theta_i$ such that $v(\theta_i^\prime) < v(\theta_i)$, we know that $\theta_i^\prime$ will be matched. There will be two situations: 1) $\theta_{last}^B$ is still reachable from $\theta_i^\prime$, 2) $\theta_{last}^B$ is not reachable from $\theta_i^\prime$. In the first situation, $\bar{\theta}_{min}^A$, $\theta_{last}^B$ and $\bar{\theta}_{max}^B$ of $m$ and $m^\prime$ are the same, so the payment will be the same for $\theta_i^\prime$ and $\theta_i$. In the second situation, we will have two sub-cases: a) $\bar{\theta}_{min}^A$ of $m$ is $\theta_{last}^A$ of $m^\prime$ and $\bar{\theta}_{max}^B$, $\theta_{last}^B$ are the same for both $m$ and $m^\prime$ (see the manipulation example in Figure \ref{fig:bf8} and \ref{fig:bf1_2} in another way around, i.e. ask of $4.5$ is misreported as ask of $2$), b) $\theta_{last}^B$ of $m$ is $\bar{\theta}_{max}^B$ of $m^\prime$ and $\bar{\theta}_{min}^A$ is the same for both $m$ and $m^\prime$ (see the manipulation example in Figure \ref{fig:bf10} and \ref{fig:bf7} (or Figure \ref{fig:bf3} and \ref{fig:bf7}) in another way around). Following the proof for the condition ``\textit{$\theta_{last}^B$ is not reachable from $\theta_i$}'' in the following, we conclude that $i$ cannot improve her utility by misreporting a lower valuation.
\end{itemize} 

\noindent(2) \textit{$\theta_{last}^B$ is not reachable from $\theta_i$ (e.g. ask 2 in Figure~\ref{fig:bf1_2}/\ref{fig:bf7}):}
\begin{itemize}
\item If $i$ reported $\theta_i^\prime$ instead of $\theta_i$ such that $v(\theta_i^\prime) > v(\theta_i)$ and \\$v(\theta_i^\prime) \leq \max(v(\theta_{last}^A),v(\bar{\theta}_{max}^B))$, we will get a new matching $m^\prime$. If $\theta_{last}^B$ is still not reachable from $\theta_i^\prime$ in $m^\prime$, then the payment for $\theta_i^\prime$ is the same as for $\theta_i$ (see Figure~\ref{fig:bf2} for example). If $\theta_{last}^B$ of $m$ is reachable from $\theta_i^\prime$ in $m^\prime$ and it is also the last matched bid of $m^\prime$ (i.e. $v(\theta_{last}^A) > v(\bar{\theta}_{max}^B)$), then $\theta_{last}^A$ of $m$ is $\bar{\theta}_{min}^A$ of $m^\prime$ and therefore the payment for $\theta_i^\prime$ will be the same as for $\theta_i$ (e.g. Figure~\ref{fig:bf1_2} and \ref{fig:bf8}). 
If $v(\theta_{last}^A) \leq v(\bar{\theta}_{max}^B)$ and $\theta_{last}^B$ of $m^\prime$ is reachable from $\theta_i^\prime$ in $m^\prime$, then $\theta_{last}^B$ of $m^\prime$ will be $\bar{\theta}_{max}^B$ of $m$ and $\theta_{last}^A$ of $m^\prime$ is either $\theta_{last}^A$ of $m$ or $\theta_i^\prime$ (see Figure~\ref{fig:bf7}, \ref{fig:bf10} and \ref{fig:bf3} for example). It is easy to check that the payment in this case is also not improved. However, if $v(\theta_i^\prime) > \max(v(\theta_{last}^A),v(\bar{\theta}_{max}^B))$, then $\theta_i^\prime$ will not be matched in $m^\prime$. Therefore, $i$ cannot improve her payment by reporting a higher valuation.
\item If $i$ reported $\theta_i^\prime$ instead of $\theta_i$ such that $v(\theta_i^\prime) < v(\theta_i)$, the ranking position of $\theta_i^\prime$ might be lower than that of $\theta_i$, but it will not change the probability for $i$ to be matched, $\theta_{last}^A$ and $\bar{\theta}_{max}^B$ are still the same, and $\theta_{last}^B$ is still not reachable from $\theta_i^\prime$. Thus, the payment will be the same for $i$ with both reports $\theta_i$ and $\theta_i^\prime$.
\end{itemize}

For an unmatched seller $i$ with ask $\theta_i$, we know that $v(\theta_{last}^A) \leq v(\theta_i) > v(\bar{\theta}_{max}^B)$. If $i$ reported $\theta_i^\prime$ such that $v(\theta_i^\prime) < v(\theta_i)$ and $\theta_i^\prime$ is matched in the new matching $m^\prime$, then there will be three cases: (a) $\theta_{last}^A$ and $\theta_{last}^B$ of $m$ are also those of $m^\prime$, (b) $\theta_{last}^A$ of $m$ is $\bar{\theta}_{min}^A$ of $m^\prime$ and $\theta_{last}^B$ of $m$ is $\theta_{last}^B$ of $m^\prime$, (c) $\bar{\theta}_{max}^B$ of $m$ is $\theta_{last}^B$ of $m^\prime$ and either $\theta_i^\prime$ or $\theta_{last}^A$ of $m$ is $\theta_{last}^A$ of $m^\prime$. For any of these three cases, the payment for $i$ with report $\theta_i^\prime$ will be less than or equal to $v(\theta_i)$, i.e. $i$ gets non-positive utility by misreporting.
\qed
\end{proof}

\begin{figure}[ht]
   \centering
   \begin{minipage}[b]{.24\linewidth}
   \centering
   \includegraphics[height=0.135\textheight]{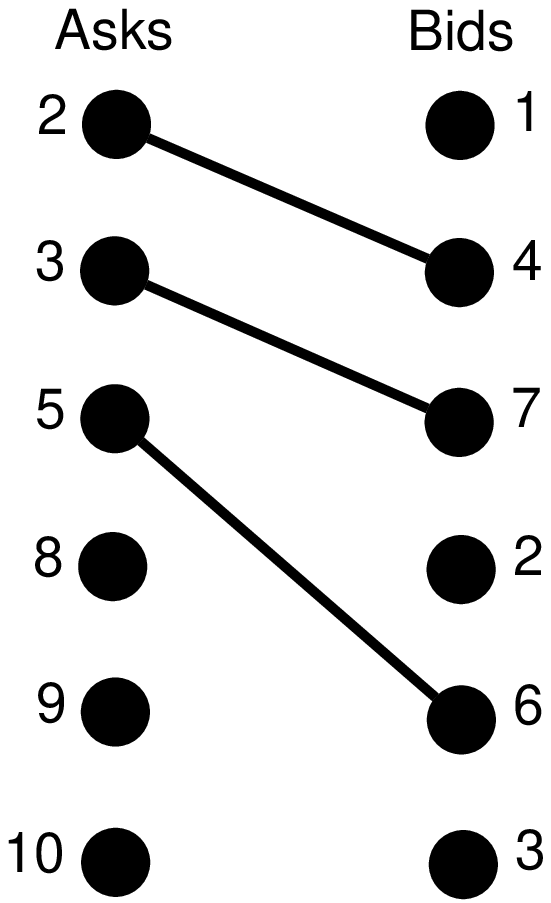}
   \subcaption{$m$}\label{fig:bf4}
   \end{minipage}    
   \begin{minipage}[b]{.24\linewidth}
   \centering
   \includegraphics[height=0.135\textheight]{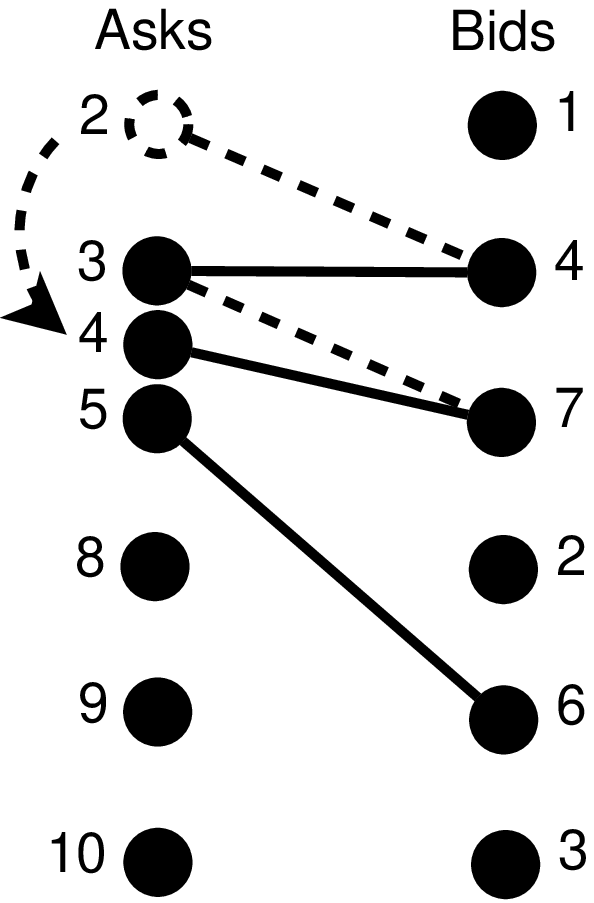}
   \subcaption{$m^\prime$}\label{fig:bf5}
   \end{minipage}
   \begin{minipage}[b]{.24\linewidth}
   \centering
   \includegraphics[height=0.135\textheight]{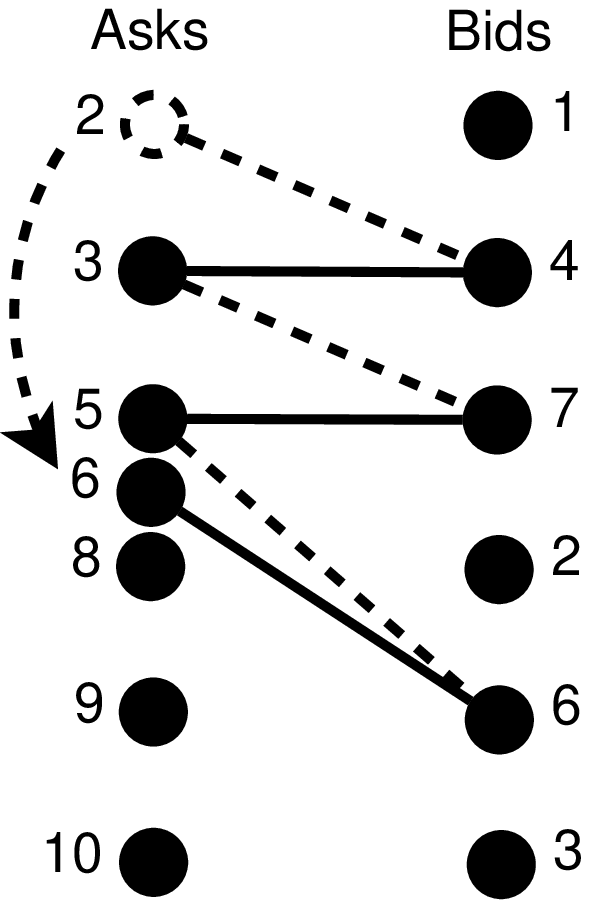}
   \subcaption{$m^{\prime\prime}$}\label{fig:bf6}
   \end{minipage}
   \begin{minipage}[b]{.24\linewidth}
   \centering
   \includegraphics[height=0.135\textheight]{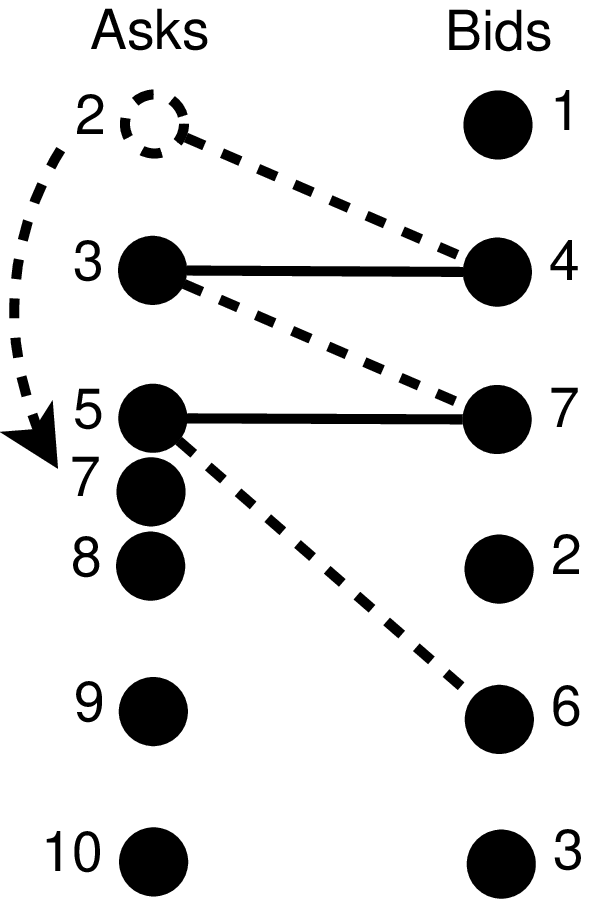}
   \subcaption{$m^{\prime\prime\prime}$}\label{fig:bf9}
   \end{minipage}
   \caption{Seller Manipulation Examples I}
   \label{fig:manipulation1}
\end{figure}

\begin{figure}[ht]
   \centering
   \begin{minipage}[b]{.24\linewidth}
   \centering
   \includegraphics[height=0.135\textheight]{./bf1}
   \subcaption{$m_1$}\label{fig:bf1_2}
   \end{minipage}
   \begin{minipage}[b]{.24\linewidth}
   \centering
   \includegraphics[height=0.135\textheight]{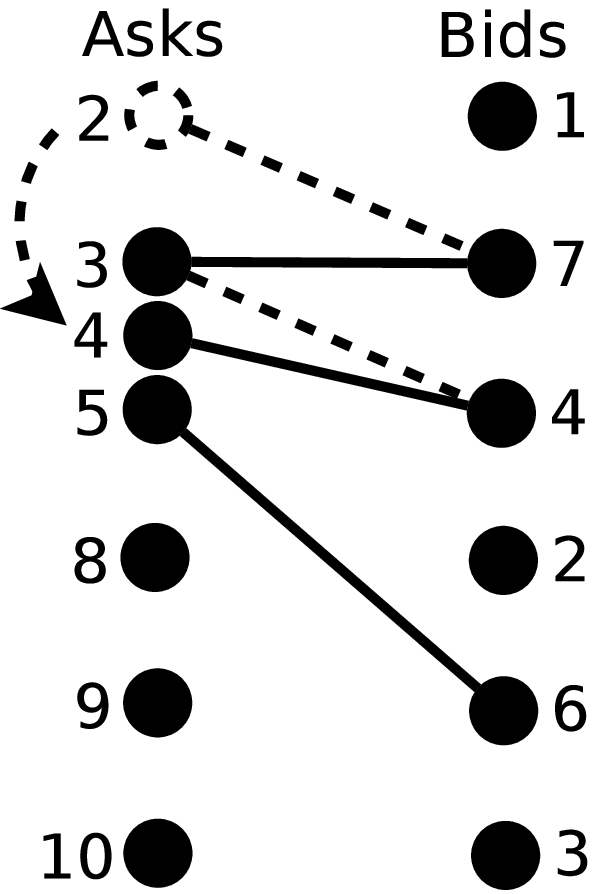}
   \subcaption{$m_1^\prime$}\label{fig:bf2}
   \end{minipage}
   \begin{minipage}[b]{.24\linewidth}
   \centering
   \includegraphics[height=0.135\textheight]{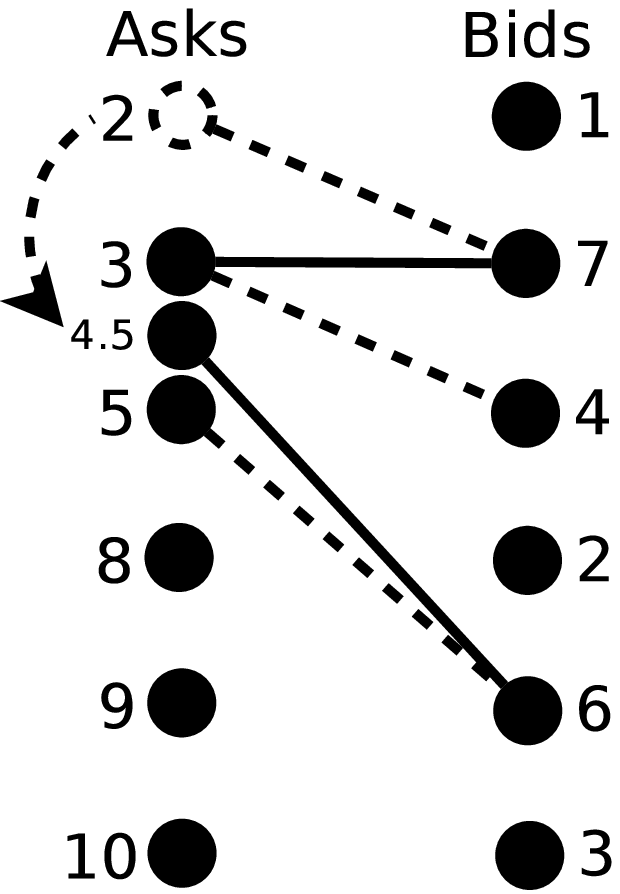}
   \subcaption{$m_1^{\prime\prime}$}\label{fig:bf8}
   \end{minipage}
\\
   \begin{minipage}[b]{.24\linewidth}
   \centering
   \includegraphics[height=0.135\textheight]{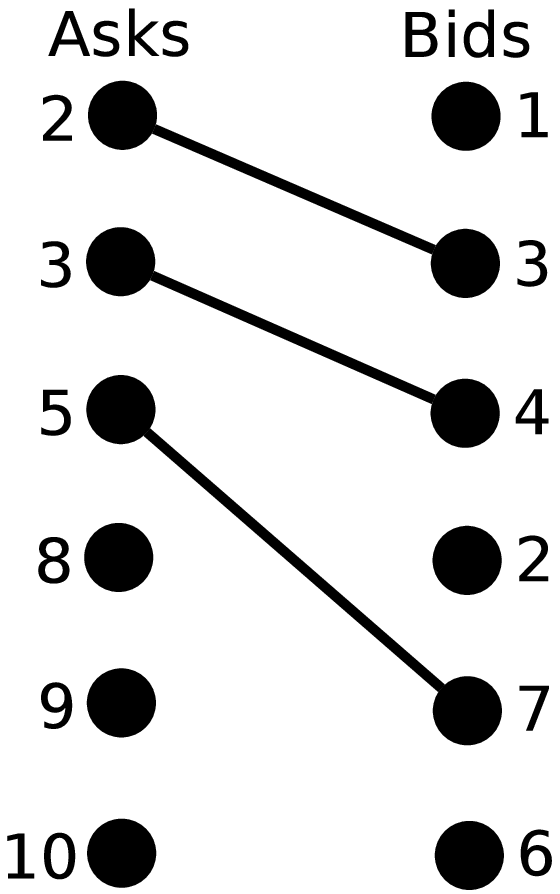}
   \subcaption{$m_2$}\label{fig:bf7}
   \end{minipage}
   \begin{minipage}[b]{.24\linewidth}
   \centering
   \includegraphics[height=0.135\textheight]{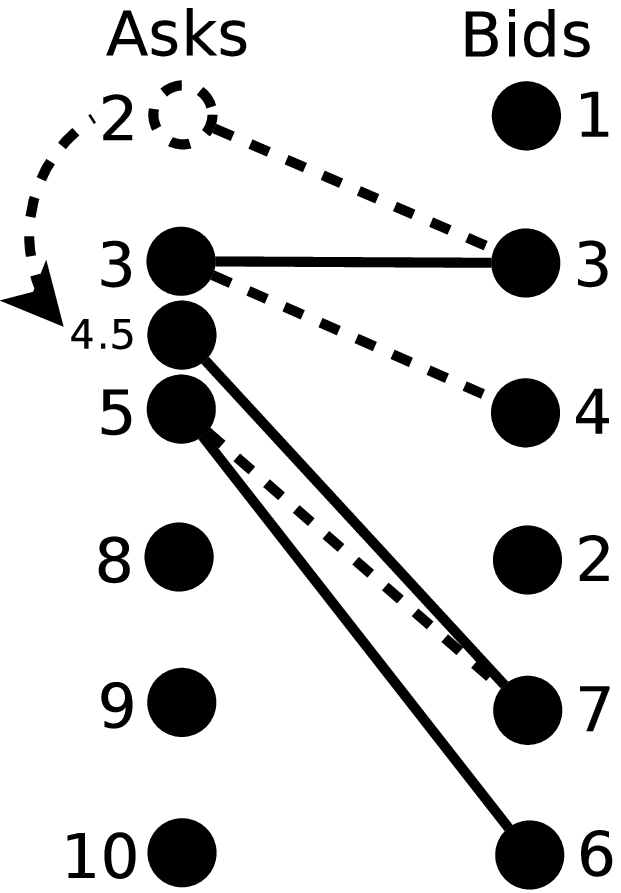}
   \subcaption{$m_2^\prime$}\label{fig:bf10}
   \end{minipage}
   \begin{minipage}[b]{.24\linewidth}
   \centering
   \includegraphics[height=0.135\textheight]{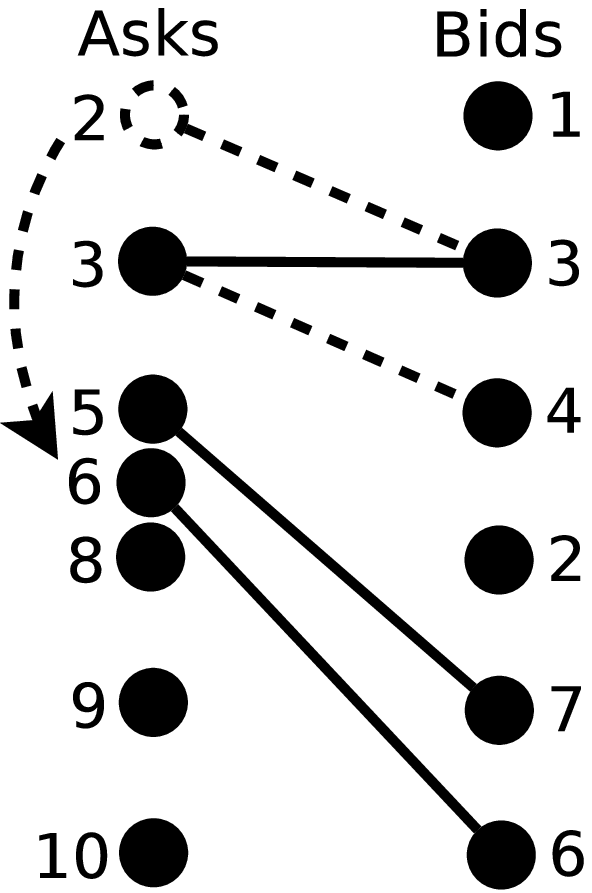}
   \subcaption{$m_2^{\prime\prime}$}\label{fig:bf3}
   \end{minipage}
   \caption{Seller Manipulation Examples II}
   \label{fig:manipulation2}
\end{figure}

\subsection*{An Extension to General Settings by Decomposing the Market}

We have assumed that sellers are patient for the proposed mechanisms $\mathcal{M}_{greedy}$ and $\mathcal{M_A}$. We demonstrate in the following how to apply them in settings without this assumption. We tackle the environment where both sellers and buyers arrive and depart randomly and one side's active time is relatively longer than the other side. Without loose of generality, we assume that sellers stay longer in the market, and each seller is active in the market for at least a period of length $t$. Assume that the whole market runs for a length of time $T$. The following mechanism, called $\mathcal{E_M}$, decomposes the whole online market into multiple sub-markets where each sub-market runs a length of time $\frac{t}{2}$.

\begin{framed}
\noindent\textbf{$\mathcal{E_M}$ of $\mathcal{M}$}\\
\rule{\textwidth}{0.5pt}
 \begin{enumerate}
  \item Split the market into $\lceil\frac{2T}{t}\rceil$ sub-markets where each sub-market $k \in \{1,2,...,\lceil\frac{2T}{t}\rceil\}$ runs in the period of $[(k-1)\frac{t}{2}, \min(k\frac{t}{2},T)]$.
  \item On the arrival of seller $i$, allocate $i$ to the latest sub-market where she is active during the whole running period of that sub-market.
  \item On the arrival of buyer $j$, allocate $j$ to all sub-markets where he can active until he is matched.
  \item Apply $\mathcal{M}$ in each sub-market.
\end{enumerate}
\end{framed}
Note that, on the arrival of a seller, she can only be allocated to a sub-market where she is active over the whole running period of that sub-market, in order to apply the proposed mechanisms. That is, she cannot be allocated to a sub-market where she arrives/departs during the sub-market is running. Also the decomposition guarantees that each seller is able to fully active in at least one sub-market.

Depending on the market situation, we will choose either $\mathcal{M}_{greedy}$ or $\mathcal{M_A}$ to instantiate the above extension. If we know that the demand is not more the supply in each sub-market, we can choose $\mathcal{M}_{greedy}$ to get the extension $\mathcal{E}_{\mathcal{M}_{greedy}}$. If the number of incoming buyers is always predictable, we can apply the reduction $\mathcal{M_A}$ in the extension. 

One interesting result we get from this extension is that when the market is rising, we get the same truthfulness and competitiveness for $\mathcal{E_M}$ as that of $\mathcal{M}$. We say a market is rising if the transaction prices are increasing, and in a good economic situation, most markets are rising, e.g. real estate.

\begin{theorem}
Under a rising market situation, if $\mathcal{M}$ is truthful and $c$-competitive, then $\mathcal{E_M}$ is truthful and $c$-competitive.
\end{theorem}

\begin{proof}[proof sketch]
For truthfulness, we know that once a trader is allocated to one sub-market, there is no beneficial manipulate in the sub-market because of the truthfulness of $\mathcal{M}$. Since the market is a rising market, so it is better for a seller to sell her item as late as possible and for a buyer to buy the item as early as possible, which is exactly what $\mathcal{E_M}$ does.

For efficiency, earlier arrival sellers with relatively lower valuations have been delayed for exchange as much as possible by $\mathcal{E_M}$ so that they can be matched to buyers with relatively higher valuations. \qed
\end{proof}

It is worth mentioning that the above mechanism actually reflects one intuitive reasoning we have for trading in a rising market situation. More specifically, in a rising market situation, a seller should sell her item as late as possible and a buyer should buy the item as early as possible in order to gain higher profit. 

Similarly, we can modify the above extension to fit for a falling/stable market situation. Moreover, how this extension can be further generalized for other market situations is worth further investigation.

\end{document}